\documentclass[10pt,journal,compsoc]{IEEEtran}
\usepackage{amsmath}
\usepackage{amsthm}
\usepackage{amssymb}
\newtheorem{thm}{Theorem}[]
\newtheorem{defn}[]{Definition}
\newtheorem{lemma}{Lemma}[]

  \usepackage{cite}
\usepackage{graphicx}
\DeclareGraphicsExtensions{.png}
\ifCLASSINFOpdf
\else
\fi
\begin{document}
%
\title{General Framework for Evaluating Password Complexity and Strength\thanks{This work is sponsored by the Department of the Air Force under Air Force Contract $\#$FA8721-05-C-0002.
Opinions, interpretations, conclusions and recommendations are those of the author and are not necessarily endorsed by the United States Government.}}
%
%
%
%

\author{Cem \c{S}. \c{S}ahin$^{\text\dag}$\thanks{$^{\text\dag}$These authors have contributed equally to this work}, Robert Lychev$^{\text\dag}$, and Neal Wagner
        \\ \small{MIT Lincoln Laboratory}
        \\ \small{244 Wood St. Lexington, MA 02421}
        \\ \small{email:{cem.sahin,robert.lychev,neal.wagner@ll.mit.edu}}
}

%
%

\markboth{\c{S}ahin, Lychev, Wagner:General Framework for Evaluating Password Complexity and Strength}{}
%



\IEEEtitleabstractindextext{%
\begin{abstract}

Although it is common for users to select \textit{bad} passwords that can be easily cracked by attackers and many decades have been invested in researching alternative authentication methods, password-based authentication remains arguably the most widely-used method of authentication today. 
Even if password use ever becomes negligible in the future, password-based authentication is currently so pervasive 
that the transition to new authentication methods is likely be very long and complicated. 
Until that happens, to encourage users to select \textit{good} passwords, enterprises often enforce policies,
for example, by requiring passwords to meet minimal length and contain special characters. 
Such policies have been proven to be ineffectual in practice, and despite all the available tools and related ongoing research, 
stolen user credentials are often cracked by attackers before victims get a chance to react properly. 
Also, researchers and practitioners often use the notions of password complexity and strength interchangeably, which only adds to the confusion users may have with respect to password selection. 
Accurate assessment of a password's resistance to cracking attacks is still an unsolved problem, and our work addresses this challenge. 
Although the best way to determine how difficult it may be to crack a user-selected password is to check its resistance to cracking attacks employed by attackers in the wild, implementing such a strategy at an enterprise would be infeasible in practice. 
In this report we, first, formalize the concepts of password complexity and strength with concrete definitions which emphasize their differences. Our definitions are quite general. 
They capture human biases and many known techniques attackers use to recover stolen credentials in real life, 
such as brute-force attacks, mangled wordlist attacks, as well as attacks that make use of Probabilistic Context Free Grammars, Markov Models, and Natural Language Processing. 
Building on our definitions, we develop a general framework for calculating password complexity and strength that could be used in practice. 
Our approach is based on the key insight that an attacker's success at cracking a password must be defined by its available computational resources, time, function used to store that password, 
as well as the \textit{topology} that bounds that attacker's search space based on that attacker's 
available inputs (prior knowledge), transformations it can use to tweak and explore its inputs, and the path of exploration which can be based on the attacker's perceived probability of success. 
We also provide a general framework for assessing the accuracy of password complexity and strength estimators that can be used to compare other tools available in the wild. 
Finally, we discuss how our framework can be used to assess procedures that rely on storing password-protected information.
\end{abstract}

}

\maketitle

\IEEEdisplaynontitleabstractindextext

%
\IEEEpeerreviewmaketitle

\IEEEraisesectionheading{\section{Introduction}\label{sec:introduction}}

Although many ideas have been proposed to replace passwords, they are still considered to be the standard authentication mechanism for such services as e-mail, 
social networking, etc.
However, password-based authentication has been a notable weak point in cyber security despite decades of effort. 
For example, in 2012, $76\%$ of network intrusions exploited weak or stolen credentials (i.e., username and/or password) \cite{verizon:2012}. 
Researchers and practitioners agree that having \textit{good} passwords is critical in many applications, but users often choose 
\textit{bad} passwords  \cite{shay:2010,komanduri:2011}.
A good password should have two key properties: (i) difficult to guess by an adversary and (ii) easy to remember; 
users almost always opt for the latter rather than the former \cite{klein:1992}. 

So many services currently rely on password-based authentication that even if password use were to ever become uncommon, 
the transition to new authentication methods is expected to be long and complicated.
And so, passwords have been the focal point of many studies in recent years. 
These studies have explored a range of related topics including password cracking algorithms \cite{narayanan:2005,castelluccia:2013,weir:2009,li:2014,veras:2014,ma:2014,ur:2015}, password strength and complexity \cite{dellamico:2010,castelluccia:2012,bonneau:2012,mazurek:2013,carnavalet:2014,passfault}, user behavior with respect to password selection \cite{florencio:2007,ur:2012,weir:2010}, and password creation policies \cite{komanduri:2011,shay:2010,shay:2012,kelley:2012,shay:2014}.
In this report we focus on accurate assessment of a password's resistance to cracking attacks, a problem that we believe still remains unsolved.

As suggested in \cite{anderson:2008}, the first step in any security analysis is to define our goals and our considered threat model.
Thus, in this report we first formally define password complexity and password strength.
In our definitions, we consider an attacker whose goal is to recover a password that has been hidden by a particular protection function.

Informally, password complexity defines the usage of allowed characters, length, and symmetry of a password. 
However, in real life, attacker's success is limited by its computational resources, time, and prior knowledge as well as how the password is stored.
Password complexity does not take such details into account, so in principle it cannot provide an accurate estimate of how long it may take an attacker to crack a password. 
Note, however, that password complexity is still a good indicator of how difficult it may be to guess a password (i.e., how close it is to a random string) when information 
about how the password is protected and/or attacker's capabilities is not readily available. 
Password strength on the other hand, does take such details into account, and, thus, it is a more complete notion.  
It may, however, be very difficult to estimate in practice because it may be impossible to accurately capture changes of such parameters 
as technological advances and any additional auxiliary information available to the adversary with time. 

Both password complexity and strength require understanding of attacker's use of prior knowledge, which we express in the form of a \textit{topology} 
that bounds the attacker's search space.
A topology is defined by the attacker's knowledge about the alphabet used to create the password,  
rules that it can use to tweak and explore words created by that alphabet, 
and the exploration path of the resultant search space.
The latter can be based on the attacker's perceived probability of success.

Our definitions are general as they capture human biases and many known techniques attackers use to recover stolen credentials in real life, 
such as brute-force attacks, mangled wordlist attacks, as well as attacks that make use of Probabilistic Context Free Grammars (PCFG), Markov Models, and Natural Language Processing (NLP).

Using our definitions we develop a general framework for calculating password complexity and strength that could be used in practice.
We believe our framework provides a complete sense of security due to its extensive consideration of how attackers crack passwords in the wild.
To summarize, the contributions of our study are as follows.
\begin{enumerate}
	\item {We formalize the concepts of password complexity and password strength.}  
	\item {We propose a novel complexity measure that models current password attacks which leverage password ``topologies,'' 
			i.e. dictionaries of words together with word-mangling rules and a specification of the order in which rules are executed during an attack.}
	\item {We provide a framework for empirical evaluation of password-strength and password-complexity estimators.}
	\item{Finally, we discuss how our framework can be used in general to assess any procedure that relis on storing password-protected information.}	
\end{enumerate}

The rest of this report is organized as follows. 
We introduce important notation and definitions which are used throughout the paper in Section~\ref{sec:prelim}. 
Section \ref{sec:PasswordStrengthvsComplexity} formally defines notions of password complexity and password strength, and then discusses their key differences and implications. 
In Section \ref{sec:ruleBasedPasswordComplexity} we present the details of our rule-based approach to calculate password complexity and strength, and we conclude in Section \ref{sec:conclusion}.

\section{Preliminaries}
\label{sec:prelim}

In this section we introduce important notation and definitions that will be used throughout the paper.

We define a finite alphabet $\alpha$ to be a finite set of characters, and a password $p$ to be a finite string over $\alpha$.
We define the profinite set $\mathbb{X}(\alpha)$ of all finite strings over $\alpha$ as the set of all possible passwords over $\alpha$. 
Note that $\mathbb{X}(\{0,1\})=\{0,1\}^{^*} $. 
We next define a password-generating procedure that we call a rule.
\begin{defn}[Rule]
\label{def:rule}
A rule, denoted by $\xi$, is a function $\xi:\mathbb{X}(\alpha)\times\{0,1\}^{^*} \rightarrow 2^{\mathbb{X}(\alpha)}$ that takes as input a finite alphabet $\alpha$ together with a finite bit string $aux$ and outputs a subset of $\mathbb{X}(\alpha)$. 
\end{defn}
Here, $aux$ is any auxiliary information that can be used to describe password-policy requirements, e.g., password minimum and maximum length, usage of capital letters and numbers, etc. 

We can view $aux$ as a logical formula specifying the requirements that users have to satisfy when selecting a password. 
It follows that when at least two rules, $\xi_i,\xi_j,\dots$, are combined to produce a new rule, $\xi_k$, $aux_k$ must be interpreted as $aux_k=aux_i \wedge aux_j \wedge \dots$ (i.e., the resulting $aux_k$ should satisfy the requirements corresponding to $aux_i$ and $aux_j$ and all other $aux$'s in this combination.) 

\begin{defn}[Rule Set]
\label{def:ruleset}
A rule set with respect to a rule $\xi$, denoted by $\{x_1,x_2,\dots\}$, is the combination of infinitely countable passwords defined by a rule, $\xi:\mathbb{X}(\alpha)\times\{0,1\}^{^*} \rightarrow 2^{\mathbb{X}(\alpha)}$ . 
\end{defn}

\begin{defn}[Combination of Rules]
\label{def:unionofRules}
The combination of any finite set of rules $\Xi=\{\xi_1,\dots,\xi_k\}$ over some finite alphabets $\alpha_1,\dots,\alpha_k\subseteq\alpha$ is the union of the outputs of those rules $\cup_{i=1}^k (\xi_i(\alpha_i,aux_i))\subseteq\mathbb{X}(\alpha)$ for any auxiliary inputs $aux_1,\dots,aux_k$.
\end{defn}

It is important to emphasize that the union of rules may consist of a single rule. 
Note that characters of $aux$ do not have to come from $\alpha$. 
For simplicity, we require that $aux$ does not specify use of characters not in $\alpha$.
The simplest example of a use of a rule is to generate all possible passwords, or all possible English dictionary words with a certain maximal length as defined in $aux$. 

\begin{defn}[Permutation of Rules]
\label{defn:ruleTopology}
A permutation of any finite set of rules $\Xi=\{\xi_1,\dots,\xi_k\}$ over some finite alphabets $\alpha_1,\dots,\alpha_k\subseteq\alpha$ outputs a directed graph $G=(V \subseteq \Xi,E \subseteq \Xi \times \Xi)$ in which the edges impose a total ordering on the vertices for any auxiliary inputs $aux_1,\dots,aux_k$.
\end{defn}
Note that we use the permutation of rules and topology interchangeably. 

\begin{defn}[Generatable]
\label{defn:generatable}
A finite string $\sigma\in\mathbb{X}(\alpha)$ is generatable by a union of rules $\Xi=\{\xi_1,\dots,\xi_k\}$, if there exist alphabets $\alpha_1,\dots,\alpha_k\subseteq\alpha$ and auxiliary inputs $aux_1,\dots,aux_k$ such that $\sigma \in \cup_{i=1}^k (\xi_i(\alpha_i,aux_i))$. 
\end{defn}

Now, we define a password parsing, which is a partitioning of a password into segments. 
\begin{defn}[Parsing]
\label{defn:parsing}
A parsing of a finite string $\sigma\in\mathbb{X}(\alpha)$ is a partition of its constituent characters in $\alpha$. 
\end{defn}

We refer to the set of all parsings of a password $p$ as $[P]$. 
\begin{defn}[Parsing Function]
Parsing function $\Gamma: \Xi\times p \rightarrow [p]\subseteq [P]$ conforms a union of rules $\Xi=\{\xi_1,\dots,\xi_k\}$ on a password $p$ and returns a list of parsings of $p$. 
\end{defn}
Note that if there is no predefined rule, $\Gamma$ generates all possible parsings of $p$. 

Now, we define a protection function that can be used to transform and/or store a password as a string such that the original password may be more difficult for an attacker to recover.   
\begin{defn}[Protection Function]
Protection function F$_{\alpha}$ \footnote{We will drop the subscript when it is clear which alphabet is being considered.} is a function $F_{\alpha} : \mathbb{X}(\alpha) \rightarrow \{0,1\}^{^*}$ that takes a finite string over $\alpha$ and outputs a bit string. 
\end{defn}

\begin{defn}[Adversary]
An adversary \footnote{We use an adversary and an attacker interchangeably. } is defined as a non-deterministic algorithm.
\end{defn} 

We use $Z(n)$ is \textit{negligible} in a parameter $n$ if $\exists c,  n_0 \in \mathbb{R}^+$ such that $Z(n) < \frac{1}{c^n}$, $\forall n > n_0$. 

We denote a probability distribution with $\mathcal{X}$. 
A specific event in the distribution $\mathcal{X}$ is shown as $x$ (i.e., $x\in\mathcal{X}$.) 
The probability that an event $x$ takes a specific value, $prob_{x}$ such that $0 < prob_x \le 1$. 
The sum of $prob_x$ over all possible values of $x$ is $1$, $\sum_{i=1}^N prob_x = 1$, where $i$ represents the event index and $N$ denotes the total number of possible events in $\mathcal{X}$. 
We assume that every event $x$ in $\mathcal{X}$, $\forall x\in\mathcal{X}$, is equally probable (i.e., $\mathcal{X}$ is uniformly distributed.)  
$|\mathcal{X}|$ denotes the cardinality of $\mathcal{X}$. 

\subsection{Example}
\label{sec:prelim-example}
In order to illustrate the above definitions, we provide the following example. 
Suppose we define an alphanumeric alphabet, $\alpha=\{A-Z,a-z,0-9\}$, with a password, $p$ as a finite string over $\alpha$. 
We define three rules, $\xi_1(\alpha,aux_1)$, $\xi_2(\alpha,aux_2)$, and $\xi_3(\alpha,aux_3)$ where $aux_1$ is a bit string that specifies passwords consisting of English dictionary words with maximum length of 8 characters, $aux_2$ is a bit string that specifies passwords consisting of numeric characters with a maximum length of $4$, and $aux_3$ specifies any alphanumeric string of length $l$, where $1\leq{l}\leq{8}$.

Rule set $\xi_1=\{p_1,p_2,\dots\},\forall i\in{Z}^+$ is a subset of $\mathbb{X}(\alpha)$ that includes $1-$to$-8$ character dictionary words , e.g. \textit{hello}, \textit{Goodbye}, etc. 
Rule set $\xi_2=\{{p_1}',{p_2}',\dots\},\forall i\in{Z}^+$ is a subset of $\mathbb{X}(\alpha)$ that includes $1-$to$-4$ character strings of digits, e.g. \textit{0011}, \textit{555}, etc. 
Rule set $\xi_3=\{{p_1}'',{p_2}'',\dots\},\forall i\in{Z}^+$ is a subset of $\mathbb{X}(\alpha)$ that includes $1-$to$-8$ character strings of letters or digits, e.g. \textit{a1b2c3}, \textit{zzz3}, etc.

A combination of the above three rules $\Xi_1=\{\xi_1,\xi_2,\xi_3\}$ is the union of rule outputs, $\cup_{i=1}^3 (\xi_i(\alpha_i,aux_i))\subseteq\mathbb{X}(\alpha)$ and includes passwords from all three rule sets above. 
A permutation of $\Xi_1=\{\xi_1,\xi_2,\xi_3\}$ gives a directed graph, $G_1$ in which the edges impose an ordering of the rules, e.g. $(\xi_2,\xi_1,\xi_3)$. 
The password, $p=password$, is generatable by rule combination, $\Xi_1$, because $p\in\cup_{i=1}^3 (\xi_i(\alpha_i,aux_i))$. 
Note that password, $p'=password1$, is not generatable by $\Xi_1$ as it is not an element of the union of rule sets $\xi_1,\xi_2,\xi_3$. 

An example parsing of $p''=psword1$ is $p|s|word1$. 
A parsing function, $\Gamma(\Xi_1,p'')$, conforms the rule combination, $\Xi_1$ on $p''$ to produce a list of parsings $[p'']$. 
Example parsings from $[p'']$ include $ps|word|1$, $ps|wo|rd1$, etc. 

A protection function $F_\alpha$ can be any one-way function that inputs a string (password) and outputs a bit string from which it is difficult to recover the original input string. 
Example protection functions include common hash functions such as MD5 or SHA-1. 
Finally, an adversary is represented as a non-deterministic password guessing algorithm, e.g. a guessing algorithm which tries dictionary words up to 8 characters in length at random and, upon exhausting all such words, tries random numbers between $1$ and $1000$.

\section{Password Complexity and Strength}
\label{sec:PasswordStrengthvsComplexity}

In this section we formally define notions of password complexity and password strength, and we then discuss their key differences and implications. 

\subsection{Defining Password Complexity} \label{sec:complex-def}
Recall that a password $p$ is just a finite string of characters that come from some particular finite alphabet $\alpha$. 
We define complexity of a given password over some alphabet in the context of a set of rules.

\begin{defn}[Complexity]
\label{def:complexity}
Complexity of a password $p$, over some alphabet $\alpha$, with respect to a finite set of rules $\Xi = \{\xi_1, \xi_2, \dots, \xi_k \}$  
is defined as the size of the smallest subset of $\mathbb{X}(\alpha)$ containing $p$ that can be generated with any combination of rules in $\Xi$ over $\alpha$, with any auxiliary inputs.
If no combination of rules in $\Xi$ can generate a set that contains $p$, then $p$'s complexity is the cardinality of $\mathbb{X}(\alpha)$.
\end{defn}

Notice that this definition requires specification of an alphabet and rules. 
This is done to capture the question of how hard it may be for an attacker to guess a password with its own set of rules and dictionaries, knowing the password policy requirements used to generate that password.  
Previous entropy-based password complexity measures were not adequate because they did not provide the means for specifying the appropriate password search space based on precisely this kind of information, i.e., rules and dictionaries that attackers may be using. 
Def.~\ref{def:complexity} also captures the scenario when the attacker has no information about password policies and cannot generate the password with any of its rules and dictionaries, in which case this password may be as good as a random finite string.

\subsection{Defining Password Strength} \label{sec:strength-def}
We define password strength with respect to the following security experiment $\text{Exp}^{\text{\tiny{FAT}-sec}}$ involving an adversary $A$. 

\begin{defn}[FAT-Security Experiment]
\label{def:fat-experiment}
The inputs to FAT-experiment are an alphabet $\alpha$, a protection function $F$ associated with $\alpha$, the description of adversary A, a password $p$ over $\alpha$, and a time period $T$. 
Description of an adversary $A$ includes all of its computational resources, its rules, and any auxiliary information. 
$A$ takes as input $\alpha$, $F$, $p$, $T$, and any additional random input. 
The security experiment ends either after $A$ outputs a finite string $p'$ or the time after the experiment starts exceeds $T$, which ever comes first. 
$\text{Exp}^{\text{\tiny{FAT}-sec}}$ ($\alpha$, $F$, $A$, $T$, $p$) returns one if, within time $T$ after its start $A$ outputs a finite string $p'$ such that $F(p')$ = $F(p)$.  
$\text{Exp}^{\text{\tiny{FAT}-sec}}$ ($\alpha$, $F$, $A$, $T$, $p$) returns zero otherwise.
\end{defn}

We now define our security definition with respect to the FAT-experiment we just described.
\begin{defn}[Password FAT-Strength]
\label{def:fat-security}
We say that a password $p$ over an alphabet $\alpha$ is FAT-secure 
if over all random inputs to $A$, $\text{Exp}^{\text{\tiny{FAT}-sec}}$($\alpha$, F, A, T, $p$) returns 0 in expectation. 
\end{defn}

Note that by providing an adversary with any auxiliary information this definition captures an attacker's potential knowledge of the policies under which the input password was selected as well as how it could steal a multitude of additional protected passwords. 
The key aspect of this definition is that it requires us to consider attacker's capabilities as well as the protection function used to store the password. 
This is something that has not been captured by any previous password strength or complexity measures, with the exception of \emph{passfault}\cite{passfault}.

\emph{passfault} is Time-To-Crack (\textsc{ttc}) estimator that takes into account attacker's capabilities including rules and password protection function. 
However, it is dependent on a fixed set of rules and a fixed methodology for parsing passwords. 
Also, it does not capture attacker's order of rule application, nor does it take into account advances in technology and any additional auxiliary information. 
We describe how to address these shortcomings later in the paper.

Note that password complexity does not take into account how the password is stored, nor attackers' capabilities.  
Thus, it cannot intrinsically provide an estimate of how long it may take anyone to crack a password.
However, it is a good indicator of how difficult it may be to guess a password (i.e., how close it is to a random string) when information about protection function or attacker's capabilities is not clear 
(the most typical scenario when users are asked to select a password for a particular website).

Password strength on the other hand, does take such details into account, and, thus, it it is more complete.  
However, password strength may be very difficult to estimate realistically because it may be impossible to accurately capture changes of such parameters as technological advances 
and any additional auxiliary information available to the adversary with time. 
Any such extra information can in principle be encapsulated as auxiliary information within the description of the adversary, and we propose how this can be done later in the paper.

\begin{figure*}[htbp]
	\centering
    \includegraphics[width=6.0in]{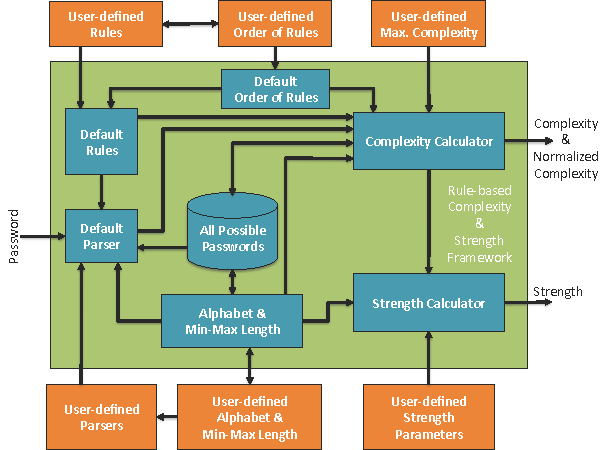}
    \caption[Our rule-based framework]{Our rule-based framework takes a password $p$ and other $aux$ as inputs and then outputs $p$'s complexity and strength. Orange is input to our framework and blue represents internal functions.}
    \label{fig:ruleBasedEngine}
\end{figure*}

\subsection{Evaluating Password Strength Estimators} \label{sec:strength-estimators-evaluation}
To truly evaluate a password-strength estimator one must compare its estimates with respect to real password-cracking attacks.
However, this may be infeasible in practice due to lack of proper equipment and time. 

Although, for the purposes of empirical evaluation, well-known password-cracking tools such 
as John the Ripper (JtR) and Hashcat~\cite{aumasson:2014} can in principle be run
on commodity hardware~\cite{yang:2015}, their performance will not do justice to attackers' capabilities in the wild.

Even when appropriate password-cracking hardware~\cite{rankin:2012} is available, there may not be enough time. 
For example, it is impractical to wait for a year to see if a password may really require that long to be recovered.
In that time attackers' capabilities are likely to improve, and the password-strength estimator under test  is likely to undergo significant updates.
To address the timing issue, one could focus on passwords that cannot be broken within a smaller, more practical amount of time $T$ (e.g., 4 weeks). 
In this context we consider the following two main criteria for evaluating password-strength estimators:
\begin{itemize}
	\item{\textbf{Reliability} The estimator does not create a false sense security in the sense that it marks weak passwords as strong.}
	\item{\textbf{Inclusion} The estimator does not reduce the space of passwords considered to be strong by marking strong passwords as weak.}
\end{itemize}

Intuitively, we do not want an estimator to overestimate or underestimate password strength.
We now present definitions that make up a framework for evaluating password-strength estimators.

\begin{defn}[FAT-Strength Estimator]
\label{def:fat-estimator}
A password FAT-strength estimator $E$ is a function that takes as input alphabet $\alpha$, a protection function $F$ associated with $\alpha$, 
the description of adversary A, a password $p$ over $\alpha$, a time period $T$, and any additional information \textsc{aux}, and outputs 
\begin{itemize}
	\item{1, in which case we say that $E$ \emph{marks} $p$ as FAT-secure, or}
	\item{0, in which case we say that $E$ \emph{marks} $p$ as not FAT-secure.}
\end{itemize} 
\end{defn}

\begin{defn}[Password FAT-Strength Estimator Reliability]
\label{def:fat-reliability}
We say that a password FAT-strength estimator $E$ over an alphabet $\alpha$ is \textbf{reliable} over a test set $P \subseteq \mathbb{X}(\alpha)$, 
if the fraction of $P$ that $E$ marks as FAT-secure that are not FAT-secure is negligible in $|P|$.
\end{defn}

\begin{defn}[Password FAT-Strength Estimator Inclusion]
\label{def:fat-inclusion}
We say that a password FAT-strength estimator $E$ over an alphabet $\alpha$ is \textbf{inclusive} over a test set $P \subseteq \mathbb{X}(\alpha)$, 
if the fraction of $P$ that $E$ marks as not FAT-secure that are FAT-secure is negligible in $|P|$.
\end{defn}

\begin{defn}[Password FAT-Strength Estimator Accuracy]
\label{def:fat-accuracy}
We say that a password FAT-strength estimator $E$ over an alphabet $\alpha$ is \textbf{accurate} over a test set $P \subseteq \mathbb{X}(\alpha)$, 
if it is both reliable and inclusive with respect to $P$.
\end{defn}

\subsection{Discussions:}
We propose a general framework using all possible types of password cracking attacks to calculate a password's complexity and strength. 
In this section, we describe a couple of examples how our framework can be used to capture different types of password cracking attacks. 

\subsubsection{Probabilistic context-free grammar:} 
In this rule, an attacker uses a large set of passwords from major password breaches to train his password generation model~\cite{weir:2009}. 
The attacker then uses the trained model to create a rule set that is used to generate a context-free grammar strings to crack a password. 
Our framework imitates the attacker's password cracking strategy during the calculation of the password's complexity and strength.

\subsubsection{Password cracking informed by online presence:}  
In this rule, an attacker scrapes the social network websites (e.g., Facebook, LinkedIn, etc.) of a user to extract possible phrases from structured/unstructured text, pictures, videos etc. which can be used to derive a password. 
These phrases can be combined with possible rule sets (e.g., word list) to create a combination of rules. 
The attacker probably knows the alphabet associated with the password that the attacker is trying to crack. 
If he does not know the alphabet, he can try different alphabets or all \textsc{ascii} characters. 
Finally, the attacker can explore all possible passwords by using different dictionaries (e.g., if the user's Facebook page has posts in English and French, the attacker can use these both dictionaries.)
Our framework mimics the attacker's password cracking strategy to calculate the password's complexity and strength.

\section{Our Rule-based Approach}
\label{sec:ruleBasedPasswordComplexity}
In this section, we present the details of our rule-based approach which uses the combinations of upper bound, lower bound, chain rule, and order-aware chain rule to calculate the complexity and strength  of a password, $p$. 
Unlike other schemes~\cite{carnavalet:2014}~\cite{castelluccia:2012}~\cite{bonneau:2012}, our framework provides more complete sense of security for a user while creating a password in $\mathbb{X}(\alpha)$. 

\subsection{Complexity and Strength Calculation Framework}
\label{subsec:framework}

General architecture of our rule-based password complexity and strength framework is shown in Fig.~\ref{fig:ruleBasedEngine}. 
It takes a password $p$ as an input. 
A user can also provide a subset or entire $\Xi$, $\alpha$, the minimum and maximum allowable length of $p$, strength parameters (e.g., password storage and adversarial capabilities), and order of rules. 
It outputs $\eta$, $FAT-$strength estimate (i.e., $0$ or $1$), and normalized complexity of $p$ in relation to lower and upper bounds of complexity. 
Note that all inputs (orange boxes in the Fig.~\ref{fig:ruleBasedEngine}) are optional. 
The \textit{default parser} extracts parsings of $p$ and it has three options: (i) using the default algorithms $p \rightarrow [P]$, (ii) using the user-defined algorithms $p\rightarrow [p]\subseteq [P]$, and (iii) locating $p$ in precalculated set of all possible passwords based on the input rules.  
The output of the default parser is either a list of parsed results of $p$ or a corresponding point of $p$ in $\mathbb{X}(\alpha)$. 
The \textit{complexity calculator} takes the default parser's output as an input. 
It may either calculate the complexity of $p$ based on the provided parsings or use pre-calculated $\Xi$ and then map $p$ into the minimum search space (i.e., complexity.) 
The complexity calculator outputs complexity and normalized complexity of $p$. 
The $FAT-$\textit{strength calculator} output is binary. 
$H_1$ means $FAT-$ strength calculator outputs $1$  and $H_0$ means $FAT-$ strength calculator outputs $0$.

\subsection{Understanding Password Complexity} 
In this section, we provide details of our framework and mathematical proofs to support our approach and also show that our rule-based complexity and strength calculation yields a sense of security for a given password $p$. 

The password entropy~\cite{principe:2010} is commonly used to indicate a measure of protection provided by $p$ and increases with the number of characters. 
The size of the all possible passwords with alphabet $\alpha$, $\mathbb{X}(\alpha)$, identifies the complexity for randomly generated passwords. 
The larger the password search space, the more difficult password is to crack by brute-force attack. 

Let us provide a couple of numerical examples to highlight some details of $\alpha$, $\mathbb{X}(\alpha)$, and adversaries perspective on $|\mathbb{X}(\alpha)|$. 
Assume that we want to create an eight-character password and the alphabet has only lower case English letters (i.e., $\xi: \alpha=\{\text{lower-case English letters}\}\times aux=\{\text{no more than eigth letters}\} \rightarrow 2^{\mathbb{X}(\alpha)}$). 
There are more than 200 billion possible ways to create a password $p$, $|\mathbb{X}{(\alpha : \text{lower case english letters})}|=26^8$. 
If an attacker knows the allowed length of $p$ and that the password only uses lower-case letters, at a rate of thousands to trillions password attempts per second, it could take $2\times10^8$  to $0.2$ seconds, respectively to crack the password by using brute force. 
Note that the rate of attempt to guess a password widely varies because of adversaries hardware capabilities. 
It might be primitive software on out-dated hardware for an everyday attacker or a dedicated infrastructure with state-of-art software algorithms for a state-sponsored cyber team. 
If we augment the alphabet with upper case letters, there are two orders of magnitude more possible ways than lower-case only passwords. 
When $\alpha$ has lower-case letters with eight-character passwords, there are $26$ times more possible ways than lower-case letters with nine-character passwords (i.e., $\alpha=\{\text{lower-case English letters}\}\times aux=\{\text{no more than nine letters}\} \rightarrow 2^{\mathbb{X}(\alpha)}$). 
These examples highlight how the size of all possible passwords, $|\mathbb{X}(\alpha)|$, changes with respect to the number of allowed length and alphabet. 
Fig.~\ref{fig:searchSpaceSize} shows various combinations of lower/upper-case letters (i.e., $\alpha=\{\text{lower- and upper-case English letters}\}\times aux=\{\text{number of letters is }\ell \in [6,15]\}\rightarrow 2^{\mathbb{X}(\alpha)}$) and password length versus the number of different ways to create a password in log scale.  

\begin{figure}[htbp]
    \includegraphics[width=3.75in]{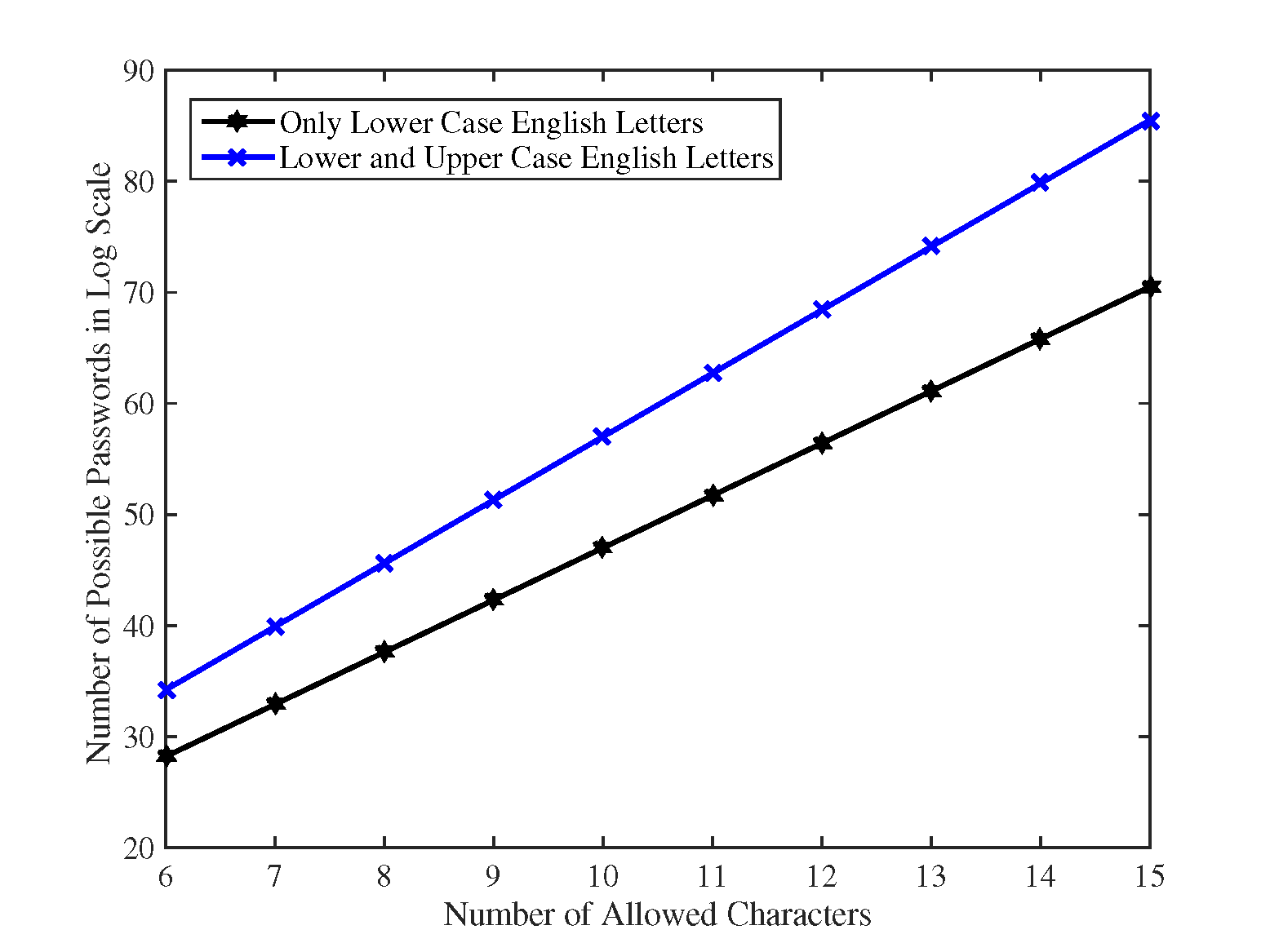}
   \caption[Complexity of $\mathbb{X}(\alpha)$ increases with the $|\alpha|$ and allowed password length]{Complexity of $\mathbb{X}(\alpha)$ increases with $|\alpha|$ and allowed password length.}
    \label{fig:searchSpaceSize}
\end{figure}

\subsection{More Information Leads to Higher Predictability}
\label{subsec:rule}
In this section, we establish mathematical ground to construct our rule-based complexity and strength calculation. 
Our goal is to prove that information accumulation about a user's password increases the predictability of the password by an adversary as we show in the following theorem. 

\begingroup
\def\thetheorem{\ref{thm:informationGain}}
\begin{thm}[Information Gain]
\label{thm:informationGain}
When an adversary $A$ gains more information about the possible password space $\mathbb{X}(\alpha)$, the complexity of a password decreases. 
\end{thm}
\addtocounter{thm}{-1}
\endgroup

To prove this theorem, we need to prove the following lemmas. 

As given in Definitions~\ref{def:rule}~and~\ref{def:ruleset}, a rule set $\xi_i=\{x_{i1},x_{i2},\dots\}$, where $\forall i\in \mathbb{Z}^+$, is a subset of $\mathbb{X}(\alpha)$. 
The following lemma shows that combination of any number of rules results in a rule which is a subset of $\mathbb{X}(\alpha)$. 

\begin{lemma}[Union]
\label{lemma:rule}
Any combinations of $n$ rules $\xi_1,\dots,\xi_n$ results in a rule $\xi_k$ whose corresponding rule set $\xi_k \text{ is } \xi_1\cup\xi_2\cup\dots\cup\xi_n \subseteq \mathbb{X}(\alpha), \forall n \in [2,\infty)$. 
%
\end{lemma}
\begin{proof}
We prove this lemma by induction,
\begin{equation}
\label{eqn:union}
\xi_1\cup\xi_2\cup\dots\cup\xi_n = \xi_k \Rightarrow \xi_k \subseteq \mathbb{X}(\alpha), \forall n \in [2,\infty)
\end{equation}

\emph{Base case:} Let $|n|=2$  such that $\xi_i \cup \xi_j = \xi_k \subseteq \mathbb{X}(\alpha)$. 
Assume that we have two rule sets $\xi_i$ and $\xi_j$ and they operate on the alphabets $\alpha_i, \alpha_j\subseteq\alpha$ and $aux_i, aux_j$, respectively. 
$\xi_i(\alpha_i,aux_i)$ and $\xi_j(\alpha_j,aux_j)$ generate $\{x_{i1},x_{i2},\dots\}$ and $\{x_{j1},x_{j2},\dots\}$, respectively. 
$\xi_i(\alpha_i,aux_i),\xi_j(\alpha_h,aux_j)\subseteq\mathbb{X}(\alpha) \Rightarrow \{x_{i1},x_{i2},\dots\}\cup\{x_{j1},x_{j2},\dots\}  \subseteq \{x_{i1},x_{i2},\dots, x_{j1},x_{j2},\dots\}$ and the cardinality of the union rule set $|\{x_{i1}, x_{i2}, \dots, x_{j1}, x_{j2},\dots\}| \le |\{x_{i1},x_{i2},\dots\}| + |\{x_{j1},x_{j2},\dots\}|$. 
Note that $|\xi_k|$ is less than and equal to $|\xi_i|+|\xi_j|=|\mathbb{X}(\alpha_i)\times aux_i| + |\mathbb{X}(\alpha_j)\times aux_j|$ since any given rule set is not necessarily proper subset of $\mathbb{X}(\alpha).$

$\xi_k$ is the combination of these two rule sets such that $\xi_i(\alpha_i,aux_i)\cup\xi_j(\alpha_j,aux_j)=\xi_k:\mathbb{X}(\alpha_i \cup \alpha_j \subseteq \alpha)\times (aux_i \wedge aux_j) \Rightarrow \xi_k= \{x_{i1},\dots, x_{ik}, x_{j1},\dots, x_{jl} \} \subseteq \mathbb{X}(\alpha)$. 
We can claim that the union of two rule sets results in a rule set in a profinite set $\mathbb{X}(\alpha)$. 

\emph{Inductive hypothesis:} 
Let $\forall n \in [2,\infty)$ be given and suppose Eqn.~\ref{eqn:union} is true for $n = q$. 
Suppose that $q$ rule sets $\xi_1,\dots,\xi_q$ operate on the alphabets $\alpha_1, \dots,\alpha_q\subseteq\alpha$ and $aux_1, \dots, aux_q$, respectively. 
$\xi_1(\alpha_1,aux_1)$, $\xi_2(\alpha_2,aux_2)$, $\dots$, and $\xi_q(\alpha_q,aux_q)$ generate $\{x_{11},x_{12},\dots\}$, $\dots$, $\{x_{q1},x_{q2},\dots\}$, respectively. 
$\xi_i(\alpha_1,aux_1),\dots,\xi_q(\alpha_q,aux_q)\subseteq\mathbb{X}(\alpha) \Rightarrow \{x_{11},x_{12}\dots\}\cup\{x_{21},x_{22},\dots\}\cup\dots\cup\{x_{q1},x_{q2},\dots\}  \subseteq \{x_{11},x_{12},\dots, x_{q1},x_{q2},\dots\}$ and the cardinality of the union rule set $|\{x_{11}, x_{12},\dots,x_{q1},x_{q2},\dots\}| \le |\{x_{11},x_{12},\dots\}| + |\{x_{21},x_{22},\dots\}|+\dots+|x_{q1},x_{q2},\dots|$. 
Note that $|\xi_k|$ is less than and equal to $|\xi_1|+|\xi_2|+\dots+|\xi_q|$ since any given rule set is not necessarily proper subset of $\mathbb{X}(\alpha).$

\emph{Induction step:} 
Let use the assumptions from Induction step - 1 and show that the result holds for $n = (q+1)$.
$\xi_1(\alpha_1,aux_1)$, $\xi_2(\alpha_2,aux_2)$, $\dots$, $\xi_{q}(\alpha_{q},aux_{q})$, and $\xi_{q+1}(\alpha_{q+1},aux_{q+1})$ generate $\{x_{11},x_{12},\dots\}$, $\dots$, $\{x_{q1},x_{q2},\dots\}$, $\{x_{(q+1)1},x_{(q+1)2},\dots\}$, respectively. 
$\xi_i(\alpha_1,aux_1)$,$\dots$,$\xi_q(\alpha_q,aux_q)$, $\xi_{(q+1)}(\alpha_{(q+1)},aux_{(q+1)})$ $\subseteq\mathbb{X}(\alpha)$ $\Rightarrow \{x_{11},x_{12}\dots\}$ $\cup\{x_{21},x_{22},\dots\}$ $\cup\dots\cup\{x_{q1},x_{q2},\dots\}$ $\cup\{x_{(q+1)1},x_{(q+1)2},\dots\}$ $\subseteq \{x_{11},x_{12},\dots, x_{q1},x_{q2},\dots,x_{(q+1)1},x_{(q+1)2},\dots\}$ and the cardinality of the union rule set $|\{x_{11}, x_{12},\dots, x_{q1},x_{q2},\dots,x_{(q+1)1},x_{(q+1)2},\dots\}| \le |\{x_{11},x_{12},\dots\}| + |\{x_{21},x_{22},\dots\}|+\dots+|x_{q1},x_{q2},\dots|+|x_{(q+1)1},x_{(q+1)2},\dots|$. 
Note that $|\xi_k|$ is less than and equal to $|\xi_1|+|\xi_2|+\dots+|\xi_q|+|\xi_{(q+1)}|$ since any given rule set is not necessarily proper subset of $\mathbb{X}(\alpha).$

\emph{Conclusion:} By the principle of induction, Eqn.~\ref{eqn:union} is true $\forall n \in [2,\infty).$
\end{proof}

Let us give an example for Lemma~\ref{lemma:rule} to give insight about the meaning of union of two rules. 
Assume that we have two rules such that $\xi_1$ represents the rule of dictionary words and $\xi_2$ is the numbers from zero to nine. 
Passwords only composed of dictionary words are in $\xi_1$ and passwords with numbers are in $\xi_2$. 
When we combine these rules, $\xi_1\cup\xi_2$, the resulting rule $\xi_3$ represents passwords with dictionary words and numbers in $\mathbb{X}(\alpha)$. 

Now, let us prove that $\mathbb{X}(\alpha)$ can be partitioned. 

\begin{lemma}[Countably Infinite]
\label{lemma:partitioned}
Profinite password search space $\mathbb{X}(\alpha)$ can be partitioned into a countably infinite set $\xi_1,\xi_2,\dots$ such that $\mathbb{X}(\alpha) = \cup_{j=1}^\infty \xi_j$ where $\xi_j$ represents a rule or a combination of rules. 
\end{lemma}
\begin{proof}
 $\mathbb{X}(\alpha)$ is defined as a set of all finite strings over an alphabet $\alpha$, 
(see Section~\ref{sec:prelim}.) 
Suppose that we have set of rules, we will now present a procedure to transform this set into a disjoint set of the rules. 
To show $\mathbb{X}(\alpha)=\cup_{j=1}^{\infty}\xi_j$, it is sufficient by producing one-to-one map $f: \cup_{i=1}^{\infty}b_i \rightarrow 2^{\mathbb{X}(\alpha)}$, where $\cup_{i=1}^{\infty}b_i$ is a pairwise disjoint countable (i.e., countably infinite) set and $\cup_{j=1}^{\infty}\xi_j = \cup_{i=1}^{\infty}b_i$ as follows:\\ 
\begin{center}
$b_1 = \xi_1$\\
$b_2 = \xi_2 - \xi_1$\\
$b_3 = \xi_3 - (\xi_1 \cup \xi_2)$\\
$.$\\
$.$\\
$.$\\
$b_n = \xi_n - (\xi_1 \cup \xi_2 \cup \dots \cup \xi_{n-1})$\\
$.$\\
$.$\\
$.$\\
\end{center}

To see $b_i$s are pairwise disjoint: let us consider $b_n$ and $b_m$, where $n < m$. 
If $x\in b_m, b_m = \xi_m - (\xi_1\cup \xi_2\cup\dots\cup \xi_{m-1})$. 
This implies that $x \notin b_n$, $b_n = \xi_n - (\xi_1 \cup \xi_2 \cup \dots \cup \xi_{n-1})$. 
Thus, $b_n \cap b_m = \varnothing$. 

Let us prove that $\cup_{i=1}^{\infty}b_i = \cup_{j=1}^{\infty}\xi_j$. 
Since $b_n \subseteq \xi_n$, we certainly have $\cup_{i=1}^{\infty}b_i \subseteq \cup_{j=1}^{\infty}\xi_j$. 
Conversely, if $x\in \cup_{j=1}^{\infty}\xi_j$, it means that $x$ is in at least one of $\xi_n$'s. 
Assume that $x\in \xi_k$, where $\xi_k$ is the smallest set that $x$ can be a member of. 
Then $x\in b_k$. 
We know that $x \notin (\xi_1\cup \dots \cup \xi_{k-1})$ so it can be discarded in the definition of $b_k$. 
$x\in\cup_{i=1}^\infty b_i$, thus,  $\cup_{i=1}^\infty b_i = \cup_{j=1}^\infty \xi_j$. 

Let us define a function $f$ such that $f:\cup_{i=1}^{\infty}b_i\rightarrow \mathbb{X}(\alpha)\Rightarrow$ if $x \in b_i$, then $f:\sigma(x)\in \mathbb{X}(\alpha)$ as given in Definition~\ref{defn:generatable}. 
For $\forall x\in\cup_{i=1}^{\infty}b_i$, there is exactly one $\sigma$ since $b_i$'s are piecewise disjoint as shown above. 
This means that there is no ambiguity in $f$'s definition. 
We can claim that $f:\sigma(x)$ is one-to-one. 
$\xi_1,\xi_2,\dots$ is countable since $\cup_{i=1}^{\infty}b_i = \cup_{i=1}^{\infty}\xi_i$. 
Thus, $\mathbb{X}(\alpha)$ can be partitioned into a countably infinite sets such that $\cup_{i=1}^{\infty}b_i=\cup_{j=1}^{\infty}\xi_j \rightarrow \mathbb{X}(\alpha)$. 
\end{proof}

Let us define the \textit{advantage} of an adversary $A$ in guessing a password $p$ as the difference between probabilities of $A$ guessing $p$ with and without prior knowledge. 

\begin{lemma} [Prior Knowledge]
\label{lemma:aprior}
If a password $p$ has a parsing $[p]\subseteq [P]$, the complexity of $p$ decreases 
when an attacker $A$ has prior knowledge. 
\end{lemma}
\begin{proof}
As shown in Lemma~\ref{lemma:partitioned}, $\mathbb{X}(\alpha)$ can be disjoint into countably infinite subsets. 
Assume that prior knowledge 
is denoted by $\mathbb{X}(\alpha')$, where $\alpha' \subseteq \alpha$. 
Regardless of the alphabets $\alpha$ and $\alpha'$, $\mathbb{X}(\alpha) \setminus \mathbb{X}(\alpha ') = \xi_k$, where $k \in \mathbb{Z}^+$
and $|\xi_k| \ge 1$. 

Let us consider a scenario in which $|\xi_k| = 1$, $\xi_k : (\alpha,aux_k)=\{p_k\}$. 
In other words, an attacker knows that $p_k \neq p$. 
Therefore, $|\mathbb{X}(\alpha')|<|\mathbb{X}(\alpha)|$ since $|\cup_{i=1}^{\infty}(\xi_i)| -  |\cup_{i=1, i\ne k}^{\infty}(\xi_i)|=|\xi_k|=1$. 

The minimal knowledge an attacker $A$ can have about a password $p$ is that $p \not\in \mathbb{X}(\alpha')$. 
Without prior knowledge, the password $p$ can be guessed with a probability of $prob_{p\in\mathbb{X}(\alpha)} = \frac{1}{|\mathbb{X}(\alpha)|}$. 
With prior knowledge (i.e., excluding $\xi_k$), the password $p$ can be guessed with a probability of $prob_{p\in\mathbb{X}(\alpha')} = \frac{1}{|\mathbb{X}(\alpha')|}=\frac{1}{|\mathbb{X}(\alpha)|-|\xi_k|}$ $\Rightarrow prob_{p\in\mathbb{X}(\alpha)} < prob_{p\in\mathbb{X}(\alpha')}$. 
Therefore, we can conclude that the complexity of $p$ decreases when there is prior knowledge. 
\end{proof}

Let us recall Theorem~\ref{thm:informationGain} which states that the smaller the search space $\mathbb{X}(\alpha)$ the smaller is the complexity of a password $p$. 

\begin{thm}[Information Gain]
\label{thm:infoGain}
The complexity of a password $p$ decreases when an adversary $A$ gains more information. 
%
\end{thm}
\begin{proof}
As shown in Lemma~\ref{lemma:aprior}, an adversary has an advantage when there is prior knowledge. 
Assume that $\mathbb{X}(\alpha')$ and $\mathbb{X}(\alpha'')$ represent subsets of $\mathbb{X}(\alpha)$ such that $\mathbb{X}(\alpha')= \cup_{i=1, i\ne k}^{\infty}(c_i)$ and $\mathbb{X}(\alpha'')= \cup_{i=1, i\ne k,l}^{\infty}(\xi_i)$, respectively. 
By definition $\alpha''\subseteq\alpha'\subseteq\alpha$. 
Regardless of the size of the subsets $\xi_k$ and $\xi_l$ and all three alphabets ($\alpha,\alpha',\alpha''$), $|\mathbb{X}(\alpha'')|<|\mathbb{X}(\alpha')|<|\mathbb{X}(\alpha)| \Rightarrow prob_{p\in\mathbb{X}(\alpha'')}>prob_{p\in\mathbb{X}(\alpha')}>prob_{p\in\mathbb{X}(\alpha)}$ as proven in Lemma~\ref{lemma:aprior}. 
Thus, we conclude that more prior information means smaller search space and the complexity of a password decreases.  
\end{proof}

\subsection{Rule-based Complexity Lower and Upper Bounds}
Password complexity depends on the password-design process and maximum complexity is achieved when each password character is independently drawn from uniformly distributed alphabet analogous to the result which shows the maximum entropy is achieved under the same conditions~\cite{shannon:1948}. 
The following lemma shows that independently drawn samples from uniformly distributed alphabet provide the maximum search space cardinality among all other distributions supporting the same alphabet. 

\begin{lemma}[Uniformly Distributed - 1]
\label{lemma:ud1}
The maximum complexity  of a password $p\in\xi(\alpha,aux)$ is obtained if and only if each character in $p$ is pulled from a uniformly distributed input set. 
\end{lemma}
\begin{proof}
As defined in Section~\ref{sec:prelim}, a password $p$ is a finite string over $\alpha$. 
Thus, let us assume that 
\begin{enumerate}
\item $p$ is composed of a finite number of strings such that $p=\{l_1,l_2,\dots,l_n\}$, where $|p|=n$ and $n\in Z^+$ is finite. 
\item The probabilities of $p=\{l_1,l_2,\dots,l_n\}$ is a set of positive real numbers $\{prob_1,prob_2,\dots,prob_n\}$, such that $l_i$ corresponds to $prob_i$. 
Note that $\sum_{i=1}^n (prob_i) = 1$. 
\end{enumerate}

We can use the inequality of arithmetic and geometric means (AM-GM inequality)~\cite{amgm:2007} to prove this lemma. 
AM-GM inequality states that the arithmetic mean of a list of non-negative real numbers is greater than or equal to the geometric mean of the same list. 
$\frac{prob_1+prob_2+\dots+prob_n}{n} \ge (prob_1\times prob_2 \dots \times prob_n)^{1/n}$
$\Rightarrow \frac{1}{n} \ge (prob_1\times prob_2 \dots \times prob_n)^{1/n}$
$\Rightarrow \frac{1}{n^{n}} \ge \frac{1}{(prob_1\times prob_2 \dots \times prob_n)}$. 
The equality holds if and only if $\forall prob_i, i\in[1,n]$ are equal. 
Thus, if each password character should be pulled from a uniformly distributed input set to obtain maximum complexity. 
\end{proof}

The following lemma show that the maximum complexity of a password $p$ can be obtained if and only if each password in $\mathbb{X}(\alpha)$ is equally likely. 

\begin{lemma}[Uniformly Distributed - 2]
\label{lemma:ud2}
The maximum complexity of a randomly created password in $\mathbb{X}(\alpha)$ is obtained if and only if $\mathbb{X}(\alpha)=\{p_1,p_2,\dots,p_n\}$ is uniformly distributed for any finite $n \in \mathbb{Z}^+$. 
\end{lemma}

\begin{proof}
We prove this lemma by using Def.~\ref{def:complexity} and Theo.~\ref{thm:infoGain}. 
The complexity is defined as a size of the smallest subset of $\mathbb{X}(\alpha)$ in Def.~\ref{def:complexity}. 
Theo.~\ref{thm:infoGain} shows that when an adversary $A$ gains more information about $\mathbb{X}(\alpha)$, the complexity of $p$ decreases. 

Let us assume that all passwords are generated by the same set of rules but $p_1$ (i.e., $\xi_1(\alpha_1\subseteq \alpha,aux_1) \rightarrow \{p_1\}$ and $\xi_2(\alpha_2\subseteq \alpha,aux_2)\rightarrow \{p_2,\dots,p_n\}$) and the probabilities of correctly guessing passwords by an adversary are $\{prob_1,prob_2,\dots,prob_n\}$, where a probability of guessing a password $p_i$ corresponds to $prob_i$. 
The relationship between the probabilities as $prob_1 \neq prob_2 = prob_3 = prob_4 = \dots = prob_n$ and $\sum_{i=1}^n(prob_i)=1$. 
The complexity of any password created by $\xi_1$ and $\xi_2$ are $|\xi_1|$ and $|\xi_2|$, respectively (see Def.~\ref{def:complexity}.) 
The complexity of $\mathbb{X}(\alpha)$, $|\mathbb{X}(\alpha)|$, is strictly greater than the complexities of $\xi_1$,$|\xi_1|$, and $\xi_2$,$|\xi_2|$, since there is an injective function, but no bijective function, from $\mathbb{X}(\alpha)$ to $\xi_1$ and $\xi_2$. 
The reason is that $p_1$ is not created by $\xi_2$ and the rest of the possible passwords (i.e., $\{p_2,\dots,p_n\}$) are not generated by $\xi_1$. 
If a password $p \not\in \xi_1$, then the probability of $p$is $\frac{1}{|\xi_2|} > \frac{1}{|\mathbb{X}(\alpha)}|$.

Thus, the maximum complexity can be obtained only all passwords in $\mathbb{X}(\alpha)$ created by a rule requiring all guessing probabilities of passwords are equally likely. 
\end{proof}

Let us first calculate the upper bound for a password $p$. 
The password-policy requirements generally provide the minimum password length, denoted by $k$ and alphabet $\alpha$. 
The maximum password length can be defined in the policy or we can get the length of the longest password, denoted by $l$, from a password database storing the existing passwords or just length of the maximum password (see Fig.~\ref{fig:ruleBasedEngine}.) 
The upper bound for our rule-based complexity is calculated as: 
\begin{equation}
\eta_{upper} = |\sum_{i=k}^{l} (\alpha)^i|
\end{equation}
$\eta_{upper}$ is $|\mathbb{X}(\alpha)|$ (i.e., the cardinality of all possible passwords) and the upper bound is the same for any rule in the set of rules and a disjoint set $\xi_i$ (see in Lemma~\ref{lemma:aprior}) of $\mathbb{X}(\alpha)$. 

The following equation provides a lower bound ($\eta_{lower/\xi}(p)$) on an adversary's effort to guess a password $p$ based on our rule-based complexity measure:

\begin{equation}
\label{eqn:lb}
  \eta_{lower/\xi}(p)=\begin{cases}
  	\begin{split}
     min(|\xi_1|,\dots,|\xi_N|,\dots), \\
      \exists |\xi_i| \text{ is bounded} \end{split}\\
    \eta_{upper} \text{ otherwise},
  \end{cases}
\end{equation}
where $i\in[ \mathbb{Z}^+$.
When a password $p$ is not part of a given rule set $\xi_i$, then for the lower bound calculation $|\xi_i|\rightarrow \infty$. 

As shown in Eq.~\ref{eqn:lb},  $\eta_{lower/\xi}(p)$ may have one of two possible outcomes. 
If a password is not a member of $\Xi$, the lower and upper bounds are equal. 
The complexity of $p$ equals to $|\xi_i|$ when the password is a part of the corresponding rule set, $\xi_i$. 
Note that $p$ can possibly be generated by more than one rule and then its complexity is the smallest cardinality of the all these rules. 

\subsection{Chain Rule Provides Complexity of Passwords Having Composite Structures}
\label{subsec:chain}
The chain rule considers the case when there is more than one pattern in a password $p$ and an adversary $A$ needs to use the combination of rules to crack the password. 

Due to the improved password-policies and richer alphabets, passwords generally have multipart structures such as combination of upper-case letters, lower-case letters, numbers, and characters. 
We define the rules as a part of $\mathbb{X}(\alpha)$ (see Definition~\ref{def:rule} and Lemmas~\ref{lemma:rule}-\ref{lemma:partitioned}) and a rule generally represents a small portion of $\mathbb{X}(\alpha)$. 
For example, if $\alpha$ is lower-case english letters and numbers from zero to nine, a rule of dictionary words only represents a small portion of $\mathbb{X}(\alpha)$.  
Thus, it is expected that passwords are generally the combination of various rules. 
The calculation of the complexity should reflect an accumulation of these small search spaces defined by the combination of rules (see Def.~\ref{def:unionofRules}) and/or a rule as shown below:
 
\begin{equation}
\label{eqn:rulechain}
  \eta_{lower/c}(p)=\begin{cases}
  \begin{split}
     min(|(c_1)|,\dots,|c_r|,\dots), \\
      \text{if } \exists |c_r| \text{ is bounded} \end{split} \\
    \eta_{upper} \text{otherwise}, 
\end{cases}
\end{equation}
where $c$ represents all possible combinations of rules in $\Xi$ and  $r\in \mathbb{Z}^+$. 
We assume that if a password $p$ is not a member of  a subset $c_i$, then $|c_i|\rightarrow \infty$, where $i \in \mathbb{Z}^+$. 
For example, $c_i$ might be a rule of dictionary words (dict.) or it might be a combination of two rules such as the first character of a password is a number and the rest of a password is dictionary words (e.g., $\underbrace{\hbox{1}}_{\hbox{\tiny{Number}}}\underbrace{\hbox{Love}}_{\hbox{\tiny{Dict.}}}\underbrace{\hbox{Soccer}}_{\hbox{\tiny{Dict.}}}$)

\subsection{Password Parsing Provides More Accurate Complexity Calculation}
Brute-force (or exhaustive search) attack is the last resort for cracking a password since it is the least efficient method. 
It requires to systematically try all the combinations. 
Brute-force always cracks a password when there is no time constraint. 
Furthermore, if a password has a predictable structure, it makes exhaustive search feasible. 
As explained up to here, our main goal is to provide better sense of security for a user. 
Therefore, we want to be conservative with the rule-based complexity calculation. 
To provide better feedback to a user, our rule-based complexity engine parses the given password to extract various patterns. 

A user can use default and/or a user-defined password parsing mechanism in order to extract patterns in a given password (see Fig.~\ref{fig:ruleBasedEngine}.) 
The parser uses an alphabet, which might be the parser specific or the common alphabet that our rule-based password complexity engine uses, to extract the patterns in a given password. 
Assume that a parser uses lower and upper case english letters and digits from zero to nine as an alphabet ($\omega = 26+26+10 = 62$) to extract the patterns with number only and three consecutive letters in a password. 
The input \textit{1LoveSoccer} can be parsed as a number of different ways such as \textit{1-Love-Soccer, 1-Lov-eSoccer, 1-LoveS-occer,} etc. 
Our rule-based complexity engine compares all these parsed results with a given rule and find the minimum search space to calculate the complexity of the password. 
For example, if we have two separate rules, namely digits ($\xi_1$) and $20K$ dictionary words ($\xi_2$), to check these parsed passwords, all extracted patterns are compared to these rules to calculate password complexity. 
Let us look at password \textit{1-Love-Soccer}. 
$\xi_1$ provides $10$ for \textit{1} (see Section~\ref{subsec:chain}), and $\infty$ for both \textit{Love} and \textit{Soccer} and  $\xi_2$ gives $\infty$ for \textit{1}, and $20K$ for both \textit{Love} and \textit{Soccer}. 
For the purpose of readability, we use $log$ scale for the following calculations. 
Thus, $\rho_{(I-Love-Soccer)}=\log_2(10\times20K\times20K)=31.8974$. 
 Now, let us calculate the complexity for \textit{1-Lov-eSoccer}. $\xi_1$ provides $\infty$ for all \textit{1}, \textit{Lov}, and \textit{eSoccer} and $\xi_2$ gives $10$ for \textit{1}, and $\infty$ for both \textit{Lov} and \textit{eSoccer}. 
Thus, $\rho_{(I-Lov-eSoccer)}= min(\log_2(62\times62^3\times62^7),\log_2(10\times62^3\times62^7)) = 65.4962$ (see Section~\ref{subsec:chain}.) 
After calculating all parsed result, the complexity of the password \textit{ILoveSoccer} is $min=(\rho_{(I-Love-Soccer)},\rho_{(I-Lov-eSoccer)},\dots)=\rho_{(I-Love-Soccer)}=31.8974$. 

\subsection{Order-Aware Chain Rule Complexity}
\label{subsec:awareOrder}
In brute-force cracking, an attacker tries every possible string in $\mathbb{X}(\alpha)$ until it succeeds. 
More common methods of password cracking, such as dictionary attacks, pattern checking, word list substitution, etc. attempt to reduce the number of trials required and will usually be attempted before exhastutive search. 
In other words, there are probable paths that an attacker can try to recover a password before trying all combinations in $\mathbb{X}(\alpha)$. 
If we have an idea of these probable paths such as an attacker checks dictionary words before word list substitution, our rule-based engine can incorporate this information into the complexity calculation as shown below: 
\begin{equation}
\label{eqn:awareOrder}
\eta(p) = min(\eta_{upper}, (|\xi_i|+|\xi_j|+\dots+|\xi_k|))
\end{equation} 
where $p$ is generatable by $\xi_k$ and an attacker $A$ tries $\xi_i$ before $\xi_j$, $\xi_j$ before $\xi_k$ and so on. 
$\forall \xi_i\in\mathbb{X}(\alpha), i \in \mathbb{Z}^+$. 

Fig.~\ref{fig:topology} presents an example of a permutation of rules. 
In this scenario, the directed graph has three nodes and two edges, $G=(V_1,E_1),\text{ where } V_1 = (\xi_i,\xi_j,\xi_k), E_1=(a_{ij},a_{jk}),\text{ and } a_{ij}=(\xi_i,\xi_j), a_{jk}=(\xi_j,\xi_k)$.   
As defined in Def.~\ref{defn:ruleTopology} and formulated in Eq.~\ref{eqn:awareOrder}, the order of evaluation of a password $p \in \mathbb{X}(\alpha)$'s complexity is $\xi_i$, $xi_j$, and then $\xi_k$. 

When there is no idea about the order of rules, the complexity calculations can use the minimum of all possible orders to provide a lower bound to a user as shown below:
\begin{equation}
\label{eqn:awareOrderMin}
\eta_{lower}(p) = min\{\eta_{upper},|\binom{\Xi}{i}|_{i=[1,\infty)} \}
\end{equation} 

\begin{figure}[htbp]
    \includegraphics[width=3.50in]{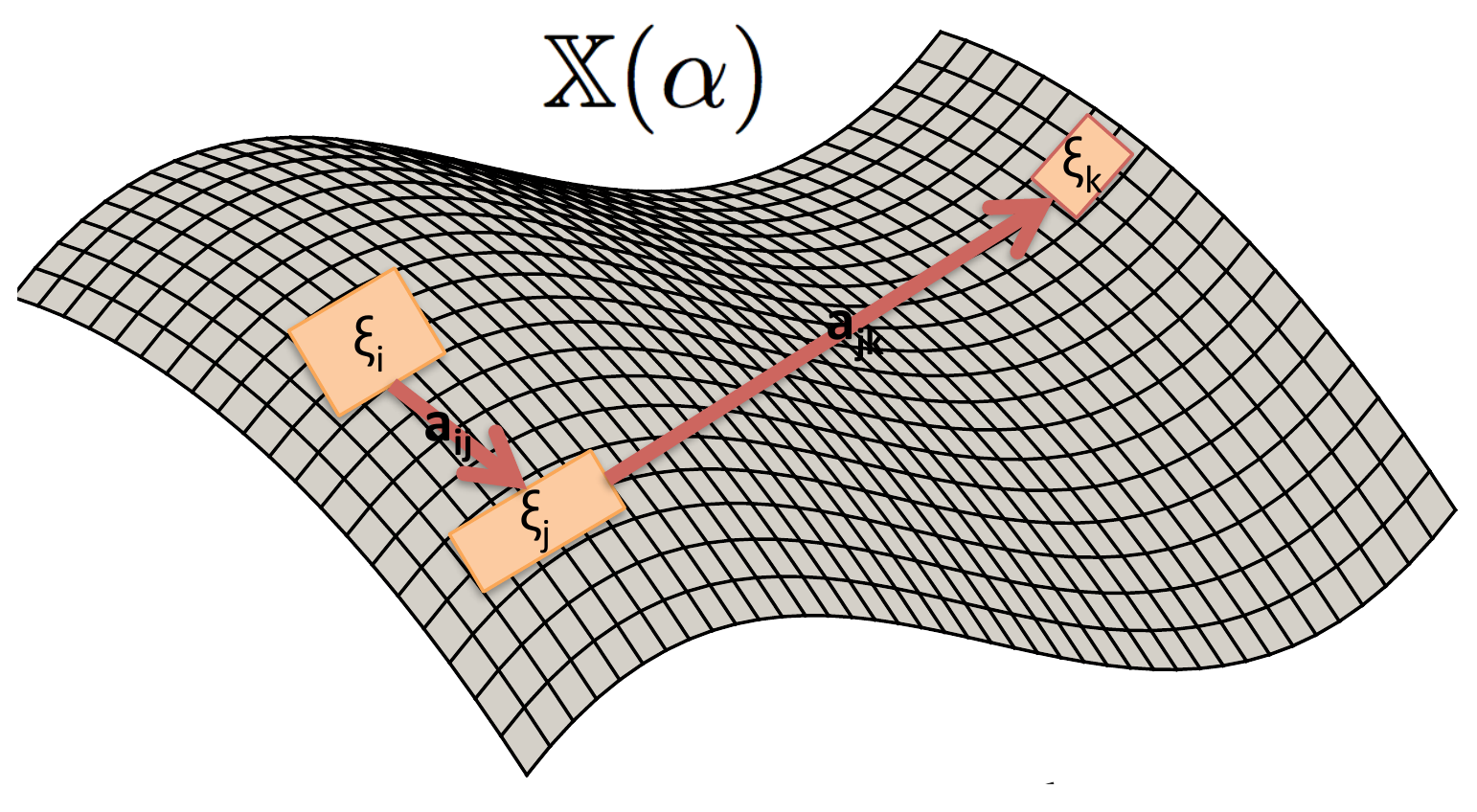}
    \caption[Topology]{An example of a permutation of rules (see Defn.~\ref{defn:ruleTopology}) shown as a directed graph $G$ such that $G=(V_1,E_1),\text{ where } V_1 = (\xi_i,\xi_j,\xi_k), E_1=(a_{ij},a_{jk}),\text{ and } a_{ij}=(\xi_i,\xi_j), a_{jk}=(\xi_j,\xi_k)$.}
    \label{fig:topology}
\end{figure}

\subsection{Password FAT-Strength}
Password FAT-strength is a calculation of the effectiveness of a password in resisting guessing and brute-force attacks. 
To ensure an acceptable level of security, our framework provides FAT-strength of a password defined in Def.~\ref{def:fat-security}. 

Most of the password strength meters categorize a password as very weak, weak, strong, and very strong~\cite{inglesant:2010}. 
They do not use the estimated time-to-crack (except passfault), an adversary's computational power, or the user's online presence. 
However, the estimation of a password strength should be a function of endurance to brute-force attack. 
Our hypothesis given in Eq.~\ref{eqn:strength} uses the factors that can be used by an adversary to calculate a password's FAT-strength. 
For example, one or more rules can be extracted from a user's online presence (e.g., facebook account). 
When a user creates a password using personal information that is publicly available, our framework has the ability to incorporate this customized information into the set of rules $\Xi$. 

If $p\in\xi_i$, then our FAT-strength calculation will include this information in the cardinality of the complexity as given below: 

\begin{equation}
\label{eqn:strength}
F \times s(t) \times \frac{1}{\mu_p(t)}\times |\eta| \underset{H_0}{\overset{H_1}{\gtrless}} T
 \end{equation}
 where $F$ is a function of a type of password storage (e.g., one-way hash function), $s(t)$ is the computation power used by an adversary to crack a password $p$, $\mu_p(t)$ is the number of parallel processors, and $T$ shows the acceptable time-to-crack that can be defined by a user or calculated from password change policy. 
 $H_{1}$ represents $FAT$-strong $p$ and $H_{0}$ represents a scenario in which $p$ is not $FAT$-strong. 
 Note that the computational power $s(t)$ is a function of time since it incorporates Moore's law (see Table~\ref{table:moore}) into account while calculation the FAT-strength.

\begin{table}[htbp]
        \caption{Expected changes in the relative computing power based on Moore's law}
        \label{table:moore}
            \begin{tabular}{|c|r|}%
                \hline
                \bf{Year} & \bf{Relative Computing Power} \\ \hline
                2015 & 1 x \\ \hline
                2025 & 32 x \\ \hline
                2035 & 1024 x \\ \hline
                2045 & 32768 x \\ \hline
            \end{tabular}
\end{table}

Fig.~\ref{fig:strength} shows a high-level model of FAT-strength calculation given in Eq.~\ref{eqn:strength}. 
Strength framework uses password protection methods, $F$, and the expected life time of $p$, $T$, as auxiliary parameters. 
Imagine that certain rules $\xi_i,\forall i\in[1,n]$ and computational power $s(t)$ can be modeled as an adversarial capabilities. 
$s(t)$ follows the Moore's law  (i.e., an adversary's computational capacity doubles every other year); however, it can also be fed into the strength framework as a different function. 
For example, if an adversary is a known-state actor and improves its computation capacity every month, this information can be incorporated into $s(t)$. 
The complexity framework provides the cardinality of an estimated complexity of $p$. 
As explained in previous sections, various number of complexities are calculated by our framework. 
$|\eta|$ is the minimum of all calculations if a user does not enforce a certain complexity calculation (e.g., order aware chain rule complexity.) 
The strength framework uses the binary test of hypothesis to decide between $H_1$ and $H_0$ which indicates whether $p$ is $FAT$-strong or not. 

\begin{figure}[htbp]
    \includegraphics[width=3.75in]{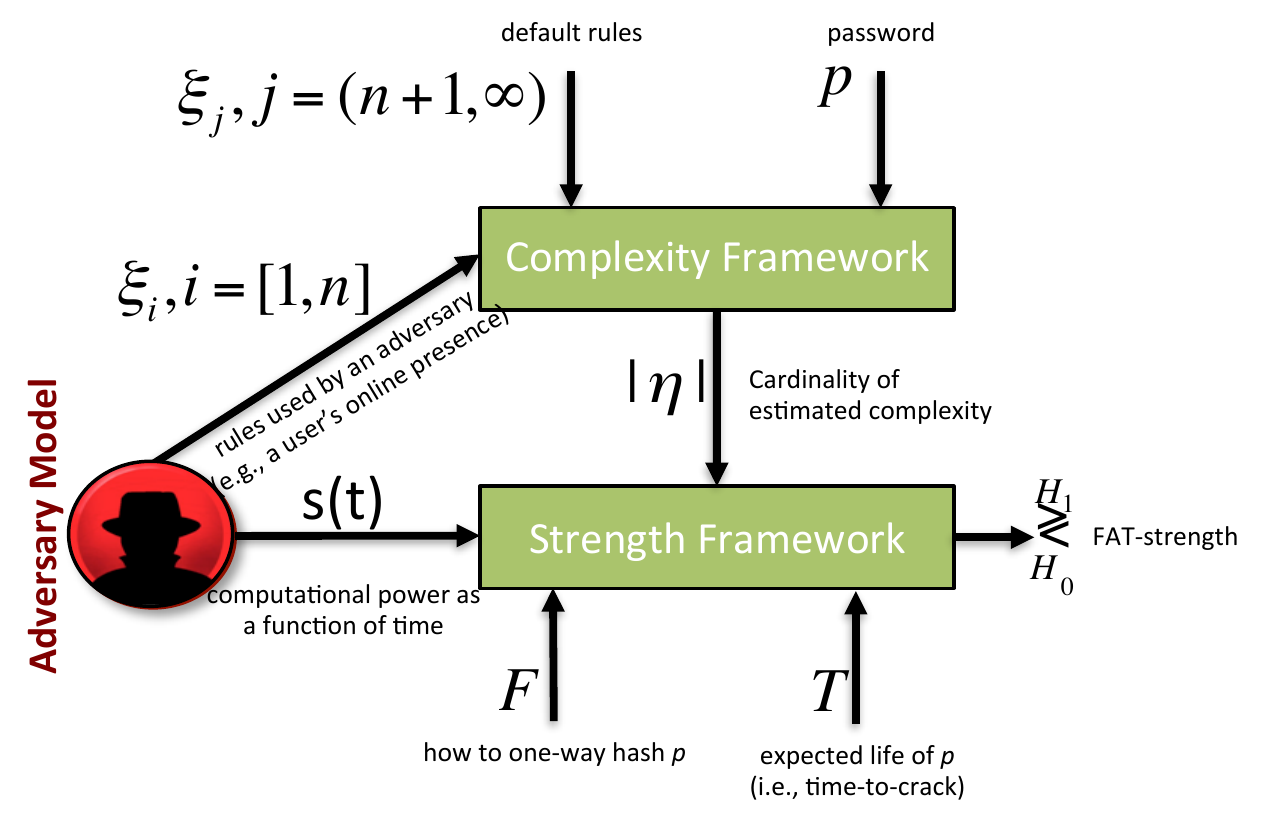}
    \caption[FAT-Strength]{FAT-strength hypothesis testing given in Eq.~\ref{eqn:strength} uses adversary model, user-defined variables (e.g., $F$ and $T$), and complexity of $p$.}
    \label{fig:strength}
\end{figure}

\section{Conclusions} \label{sec:conclusion}
In this paper we formalize the concepts of password complexity and password strength and propose a novel approach to calculate password strength and complexity while providing a general framework for analyzing/comparing other available password strength/complexity estimators. 
Our framework incorporated human biases into our calculation so that lower-bound of a password strength and complexity can be provided to a user. 
The key insight we employ is that a brute-force attacker does not assume all guesses are equally likely, so one should not assume all possible passwords are equally good. 
As a result, our framework to calculating password strength and complexity uses the idea that some guesses are far better than others since human-based password choices are not random. 
Furthermore, our approach can easily be generalized to accommodate other methods for storing secret information and authenticating identities and/or accounts. 




\bibliographystyle{IEEEtran}
\bibliography{passwords}

\begin{thebibliography}{10}
\providecommand{\url}[1]{#1}
\csname url@samestyle\endcsname
\providecommand{\newblock}{\relax}
\providecommand{\bibinfo}[2]{#2}
\providecommand{\BIBentrySTDinterwordspacing}{\spaceskip=0pt\relax}
\providecommand{\BIBentryALTinterwordstretchfactor}{4}
\providecommand{\BIBentryALTinterwordspacing}{\spaceskip=\fontdimen2\font plus
\BIBentryALTinterwordstretchfactor\fontdimen3\font minus
  \fontdimen4\font\relax}
\providecommand{\BIBforeignlanguage}[2]{{%
\expandafter\ifx\csname l@#1\endcsname\relax
\typeout{** WARNING: IEEEtran.bst: No hyphenation pattern has been}%
\typeout{** loaded for the language `#1'. Using the pattern for}%
\typeout{** the default language instead.}%
\else
\language=\csname l@#1\endcsname
\fi
#2}}
\providecommand{\BIBdecl}{\relax}
\BIBdecl

\bibitem{verizon:2012}
V.~R. Team \emph{et~al.}, ``Verizon 2012 data breach investigations report,''
  Technical report, Tech. Rep., 2012.

\bibitem{shay:2010}
R.~Shay, S.~Komanduri, P.~G. Kelley, P.~G. Leon, M.~L. Mazurek, L.~Bauer,
  N.~Christin, and L.~F. Cranor, ``Encountering stronger password requirements:
  User attitudes and behaviors,'' in \emph{Proceedings of the Sixth Symposium
  on Usable Privacy and Security}, ser. SOUPS '10.\hskip 1em plus 0.5em minus
  0.4em\relax New York, NY, USA: ACM, 2010, pp. 2:1--2:20.

\bibitem{komanduri:2011}
S.~Komanduri, R.~Shay, P.~G. Kelley, M.~L. Mazurek, L.~Bauer, N.~Christin,
  L.~F. Cranor, and S.~Egelman, ``Of passwords and people: Measuring the effect
  of password-composition policies,'' in \emph{Proceedings of the SIGCHI
  Conference on Human Factors in Computing Systems}, ser. CHI '11.\hskip 1em
  plus 0.5em minus 0.4em\relax New York, NY, USA: ACM, 2011, pp. 2595--2604.

\bibitem{klein:1992}
D.~V. Klein, ``Foiling the cracker: A survey of, and improvements to, password
  security,'' in \emph{Proceedings of the 2nd USENIX Security Workshop}, 1990,
  pp. 5--14.

\bibitem{narayanan:2005}
A.~Narayanan and V.~Shmatikov, ``Fast dictionary attacks on passwords using
  time-space tradeoff,'' in \emph{Proceedings of the 12th ACM Conference on
  Computer and Communications Security}, ser. CCS '05.\hskip 1em plus 0.5em
  minus 0.4em\relax New York, NY, USA: ACM, 2005, pp. 364--372.

\bibitem{castelluccia:2013}
C.~Castelluccia, C.~Abdelberi, M.~D{\"{u}}rmuth, and D.~Perito, ``When privacy
  meets security: Leveraging personal information for password cracking,''
  \emph{CoRR}, vol. abs/1304.6584, 2013.

\bibitem{weir:2009}
M.~Weir, S.~Aggarwal, B.~de~Medeiros, and B.~Glodek, ``Password cracking using
  probabilistic context-free grammars,'' in \emph{Proceedings of the IEEE
  Symposium on Security and Privacy}, May 2009, pp. 391--405.

\bibitem{li:2014}
Z.~Li, W.~Han, and W.~Xu, ``A large-scale empirical analysis of chinese web
  passwords,'' in \emph{Proc. 23rd USENIX Security Symposium, USENIX Security
  (August 2014)}, 2014.

\bibitem{veras:2014}
R.~Veras, C.~Collins, and J.~Thorpe, ``On the semantic patterns of passwords
  and their security impact,'' in \emph{Proceedings of the Network and
  Distributed System Security Symposium (NDSS'14)}, 2014.

\bibitem{ma:2014}
J.~Ma, W.~Yang, M.~Luo, and N.~Li, ``A study of probabilistic password
  models,'' in \emph{Proceedings of the IEEE Symposium on Security and
  Privacy}, May 2014, pp. 689--704.

\bibitem{ur:2015}
B.~Ur, S.~M. Segreti, L.~Bauer, N.~Christin, L.~F. Cranor, S.~Komanduri,
  D.~Kurilova, M.~L. Mazurek, W.~Melicher, and R.~Shay, ``Measuring real-world
  accuracies and biases in modeling password guessability,'' in \emph{24th
  USENIX Security Symposium (USENIX Security 15)}.\hskip 1em plus 0.5em minus
  0.4em\relax Washington, D.C.: USENIX Association, Aug. 2015, pp. 463--481.

\bibitem{dellamico:2010}
M.~Dell'Amico, P.~Michiardi, and Y.~Roudier, ``Password strength: An empirical
  analysis,'' in \emph{INFOCOM, 2010 Proceedings IEEE}, March 2010, pp. 1--9.

\bibitem{castelluccia:2012}
C.~Castelluccia, M.~D{\"u}rmuth, and D.~Perito, ``Adaptive password-strength
  meters from markov models.'' in \emph{Proceedings of the Network and
  Distributed System Security Symposium (NDSS)}, 2012.

\bibitem{bonneau:2012}
J.~Bonneau, ``The science of guessing: Analyzing an anonymized corpus of 70
  million passwords,'' in \emph{Proceedings of the IEEE Symposium on Security
  and Privacy}, May 2012, pp. 538--552.

\bibitem{mazurek:2013}
M.~L. Mazurek, S.~Komanduri, T.~Vidas, L.~Bauer, N.~Christin, L.~F. Cranor,
  P.~G. Kelley, R.~Shay, and B.~Ur, ``Measuring password guessability for an
  entire university,'' in \emph{Proceedings of the 2013 ACM SIGSAC Conference
  on Computer \&\#38; Communications Security}, ser. CCS '13.\hskip 1em plus
  0.5em minus 0.4em\relax New York, NY, USA: ACM, 2013, pp. 173--186.

\bibitem{carnavalet:2014}
X.~de~Carn{\'e}~de Carnavalet and M.~Mannan, ``From very weak to very strong:
  Analyzing password-strength meters,'' in \emph{Network and Distributed System
  Security (NDSS) Symposium 2014}.\hskip 1em plus 0.5em minus 0.4em\relax
  Internet Society, February 2014.

\bibitem{passfault}
Passfault, http://www.passfault.com/.

\bibitem{florencio:2007}
D.~Florencio and C.~Herley, ``A large-scale study of web password habits,'' in
  \emph{Proceedings of the 16th International Conference on World Wide Web},
  ser. WWW '07.\hskip 1em plus 0.5em minus 0.4em\relax New York, NY, USA: ACM,
  2007, pp. 657--666.

\bibitem{ur:2012}
B.~Ur, P.~G. Kelley, S.~Komanduri, J.~Lee, M.~Maass, M.~L. Mazurek, T.~Passaro,
  R.~Shay, T.~Vidas, L.~Bauer \emph{et~al.}, ``How does your password measure
  up? the effect of strength meters on password creation.'' in \emph{USENIX
  Security Symposium}, 2012, pp. 65--80.

\bibitem{weir:2010}
M.~Weir, S.~Aggarwal, M.~Collins, and H.~Stern, ``Testing metrics for password
  creation policies by attacking large sets of revealed passwords,'' in
  \emph{Proceedings of the 17th ACM Conference on Computer and Communications
  Security}, ser. CCS '10.\hskip 1em plus 0.5em minus 0.4em\relax New York, NY,
  USA: ACM, 2010, pp. 162--175.

\bibitem{shay:2012}
R.~Shay, P.~G. Kelley, S.~Komanduri, M.~L. Mazurek, B.~Ur, T.~Vidas, L.~Bauer,
  N.~Christin, and L.~F. Cranor, ``Correct horse battery staple: Exploring the
  usability of system-assigned passphrases,'' in \emph{Proceedings of the
  Eighth Symposium on Usable Privacy and Security}, ser. SOUPS '12.\hskip 1em
  plus 0.5em minus 0.4em\relax New York, NY, USA: ACM, 2012, pp. 1--20.

\bibitem{kelley:2012}
P.~Kelley, S.~Komanduri, M.~Mazurek, R.~Shay, T.~Vidas, L.~Bauer, N.~Christin,
  L.~Cranor, and J.~Lopez, ``Guess again (and again and again): Measuring
  password strength by simulating password-cracking algorithms,'' in
  \emph{Security and Privacy (SP), 2012 IEEE Symposium on}, May 2012, pp.
  523--537.

\bibitem{shay:2014}
R.~Shay, S.~Komanduri, A.~L. Durity, P.~S. Huh, M.~L. Mazurek, S.~M. Segreti,
  B.~Ur, L.~Bauer, N.~Christin, and L.~F. Cranor, ``Can long passwords be
  secure and usable?'' in \emph{Proceedings of the 32Nd Annual ACM Conference
  on Human Factors in Computing Systems}, ser. CHI '14.\hskip 1em plus 0.5em
  minus 0.4em\relax New York, NY, USA: ACM, 2014, pp. 2927--2936.

\bibitem{anderson:2008}
R.~Anderson, \emph{Security engineering}.\hskip 1em plus 0.5em minus
  0.4em\relax John Wiley \& Sons, 2008.

\bibitem{aumasson:2014}
J.-P. Aumasson, W.~Meier, R.~C.-W. Phan, and L.~Henzen, \emph{The Hash Function
  BLAKE}.\hskip 1em plus 0.5em minus 0.4em\relax Springer, 2014.

\bibitem{yang:2015}
S.~Yang, S.~Ji, X.~Hu, and R.~Beyah, ``Effectiveness and soundness of
  commercial password strength meters,'' 2015.

\bibitem{rankin:2012}
K.~Rankin, ``Hack and /: Password cracking with gpus, part ii: Get cracking,''
  \emph{Linux J.}, vol. 2012, no. 214, Feb. 2012.

\bibitem{principe:2010}
J.~C. Principe, \emph{Information theoretic learning: Renyi's entropy and
  kernel perspectives}.\hskip 1em plus 0.5em minus 0.4em\relax Springer Science
  \& Business Media, 2010.

\bibitem{shannon:1948}
C.~E. Shannon, ``A mathematical theory of communication,'' \emph{SIGMOBILE Mob.
  Comput. Commun. Rev.}, vol.~5, no.~1, pp. 3--55, Jan. 2001.

\bibitem{amgm:2007}
M.~Hirschhorn, ``The am-gm inequality,'' \emph{The Mathematical Intelligencer},
  vol.~29, no.~4, pp. 7--7, 2007.

\bibitem{inglesant:2010}
P.~G. Inglesant and M.~A. Sasse, ``The true cost of unusable password policies:
  Password use in the wild,'' in \emph{Proceedings of the SIGCHI Conference on
  Human Factors in Computing Systems}, ser. CHI '10.\hskip 1em plus 0.5em minus
  0.4em\relax New York, NY, USA: ACM, 2010, pp. 383--392.

\end{thebibliography}




\end{document}